\documentclass[11pt]{article}
\usepackage{amsmath,graphicx,xspace,amsfonts,amsthm,a4,natbib,amsmath,changebar,url}

\newcommand{\BI}{\begin{itemize}}
\newcommand{\EI}{\end{itemize}}
\newcommand{\BE}{\begin{enumerate}}
\newcommand{\EE}{\end{enumerate}}

\newcommand{\BC}{\begin{center}}
\newcommand{\EC}{\end{center}}
\newcommand{\be}{\begin{equation}}
\newcommand{\bel}[1]{\begin{equation}\label{#1}}
\newcommand{\ee}{\end{equation}}

\newcommand{\CI}{\perp\!\!\!\perp}

\newtheorem{theorem}{Theorem}
\newtheorem{lemma}{Lemma}
\theoremstyle{definition}
\newtheorem{definition}{Definition}
\theoremstyle{remark}
\newtheorem{remark}{Remark}
\newtheorem*{example}{Example}

\newcommand{\fU}{\mathfrak{U}}
\newcommand{\fG}{\mathfrak{G}}
\newcommand{\fF}{\mathfrak{F}}

\newcommand{\cG}{\mathcal{G}}
\newcommand{\cC}{\mathcal{C}}
\newcommand{\cS}{\mathcal{S}}

\newcommand{\cGx}[1]{\cG^{(#1)}}
\newcommand{\cGA}[1]{\cG_A^{(#1)}}
\newcommand{\cGB}[1]{\cG_B^{(#1)}}

\begin{document}
\title{{\sc A structural Markov property for decomposable graph laws that allows control of clique intersections}}
\author{
Peter J. Green\thanks {School of Mathematics, University of
Bristol, Bristol BS8 1TW, UK.
\newline \hspace*{5mm} Email: {\tt P.J.Green@bristol.ac.uk}.}\\
\\University of Technology Sydney\\ and University of Bristol.
\and 
Alun Thomas\thanks {Division of Genetic Epidemiology, Department of Internal Medicine, University of Utah.
\newline \hspace*{5mm} Email: {\tt Alun.Thomas@utah.edu}}\\
University of Utah.
}
\date{\today}
\maketitle

\begin{abstract}
We present a new kind of structural Markov property for probabilistic laws on decomposable graphs, which allows the explicit control of interactions between cliques,
so is capable of encoding some interesting structure. We prove the equivalence of this property to an exponential family assumption, and discuss identifiability, modelling, inferential and computational implications.

\hspace{5mm}

\noindent {\small {\em Some key words:} 
conditional independence,
graphical model, 
hub model,
Markov random field,
model determination,
random graphs
}

\end{abstract}

\section{Introduction}
The conditional independence properties among components of a multivariate distribution are key to understanding its structure, and precisely describe the qualitative manner in which information flows among the variables. Further, these properties are well-represented by 
a graphical model, in which 
nodes, representing variables in the model, are connected by 
undirected edges, encoding the conditional independence properties of the distribution \citep{lauritzen-96}. Inference about the underlying graph from observed data is therefore an important task, sometimes known as structural learning.

Bayesian structural learning requires specification of a prior distribution on graphs, and there is a need for a flexible but tractable family of such priors, capable of representing a variety of prior beliefs about the conditional independence structure. 
In the interests of tractability and scalability, there has been a strong focus on the case where the true graph may be assumed to be decomposable. 

Just as this underlying graph
localises the pattern of dependence among variables, it is appealing that the prior on the graph itself should exhibit dependence locally, in the same graphical sense. Informally, the presence or absence of two edges should be independent when they are sufficiently separated by other edges in the graph. The first class of graph priors demonstrating such a structural Markov property was presented in a 2012 Cambridge University PhD thesis by Simon Byrne, and later published in \citet{byrne-dawid-15}. 

That priors with this property are also tractable arises from an equivalence demonstrated by \citet{byrne-dawid-15}, between their structural Markov property for decomposable graphs and 
the assumption that the graph law follows a clique exponential family.

This important result is yet another example of a theme 
in the literature, making a connection between systems of conditional independence statements among random variables, often encoded graphically, and factorisations of the joint probability distribution of these variables. Examples include 
the global Markov property for undirected graphs, which is necessary, and under an additional condition sufficient, for the joint distribution to factorise as a product of potentials over cliques;
the Markov property for directed acyclic graphs, which is equivalent to the existence of a factorisation of the joint distributions into child-given-parents conditional distributions;
and the existence of a factorisation into clique and separator marginal distributions for undirected decomposable graphs.

All of these results are now well-known, and for these and other essentials of graphical models, the reader is referred to \citet{lauritzen-96}.

In this note, we introduce a weaker version of this structural Markov property, and show that it is nevertheless sufficient for equivalence to a certain exponential family, and therefore to a factorisation of the graph law. This gives us a more flexible family of graph priors for use in modelling data. We show that the advantages of conjugacy, and its favourable computational implications, remain true in this broader class, and illustrate the richer structures that are generated by such priors. Efficient prior and posterior sampling from decomposable graphical models can be performed with the junction tree sampler of \citet{green-thomas-13}.

\section{The weak structural Markov property}

\subsection{Notation and terminology}

We follow the terminology for graphs and graphical models of \citet{lauritzen-96}, with a few exceptions and additions, noted here. Many of these are also used by \citet{byrne-dawid-15}. We use the term graph law for the distribution of a random graph, but do not use a different symbol, for example $\tilde{\cG}$, for a random graph. For any graph $\cG$ on a vertex set $V$, and any subset $A\subseteq V$, $\cG_A$ is the subgraph induced on vertex set $A$; its edges are those of $\cG$ joining vertices that are both in $A$. A complete subgraph is one where all pairs of vertices are joined. If $\cG_A$ is complete and maximal, in the sense that $\cG_B$ is not complete for any superset $B\supset A$, then $A$ is a clique. Here and throughout the paper, the symbols $\supset$ and $\subset$ refer to strict inclusion.
A junction tree based on a decomposable graph $\cG$ on vertex set $V$ is any graph whose vertices are the cliques of $\cG$, joined by edges in such a way that for any $A\subseteq V$, those vertices of the junction tree containing $A$ form a connected subtree. A separator is the intersection of two adjacent cliques in any junction tree. As in \citet{green-thomas-13} we adopt the convention that we allow separators to be empty, with the effect that every junction tree is connected. A covering pair is any pair $(A,B)$ of subsets of $V$ such that $A\cup B=V$; $(A,B)$ is a decomposition if $A\cap B$ is complete, and separates $A\setminus B$ and $B\setminus A$. Figure \ref{fig:decompos} illustrates the idea of a decomposition.

\begin{figure}[ht]
\begin{center}
\resizebox{0.6\textwidth}{!}{\includegraphics{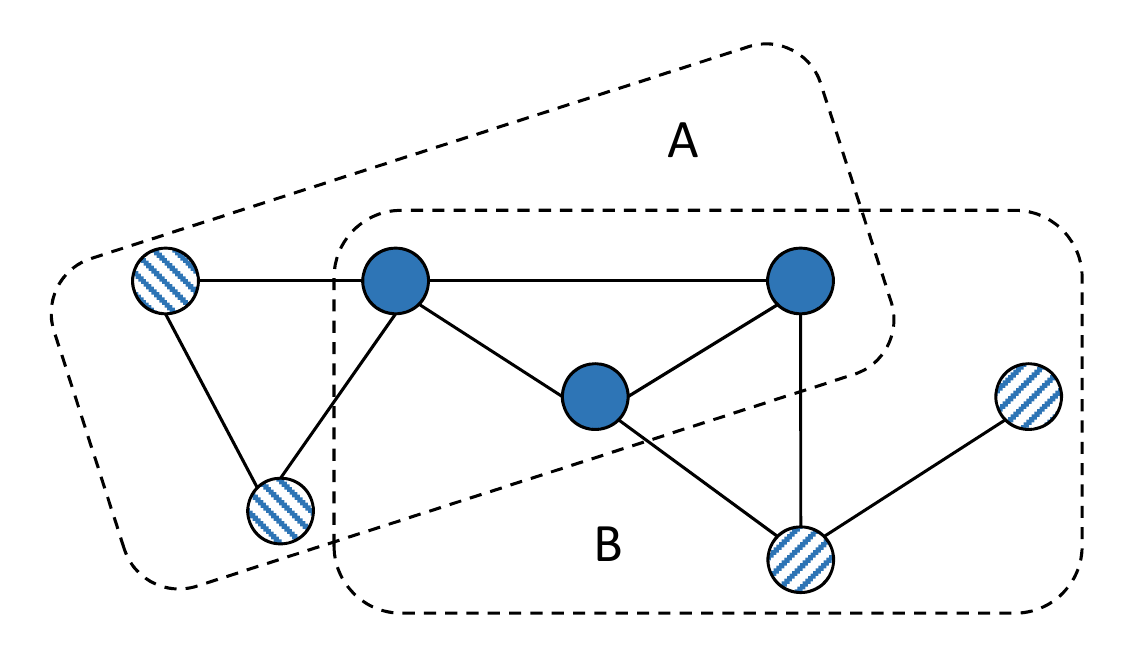}}
\end{center}
\caption{A decomposition: $A\cap B$ is complete and separates $A\setminus B$ and $B \setminus A$. \label{fig:decompos}}
\end{figure}

\subsection{Definitions}
We begin with the definition of the structural Markov property from \citet{byrne-dawid-15}. 
\begin{definition} (Structural Markov property)
A graph law $\fG(\cG)$ over the set $\fU$ of undirected decomposable graphs on $V$ is \emph{structurally Markov} if for any covering pair $(A,B)$, we have
$$
\cG_A \CI \cG_B \mid \{\cG \in \fU(A,B)\} \quad[\fG],
$$
where $\fU(A,B)$ is the set of decomposable graphs for which $(A,B)$ is a decomposition.
\end{definition}

The various conditional independence statements each restrict the graph law, so we can weaken the definition by reducing the number of such statements, for example by replacing the conditioning set by a smaller one. This motivates our definition.
\begin{definition} (Weak structural Markov property)
A graph law $\fG(\cG)$ over the set $\fU$ of undirected decomposable graphs on $V$ is \emph{weakly structurally Markov} if for any covering pair $(A,B)$, we have
$$
\cG_A \CI \cG_B \mid \{\cG \in \fU^\star(A,B)\} \quad[\fG],
$$
where $\fU^\star(A,B)$ is the set of decomposable graphs for which $(A,B)$ is a decomposition, and $A\cap B$ is a clique, that is a maximal complete subgraph, in $\cG_A$.
\end{definition}

The only difference with the structural Markov property is that we condition on the event $\fU^\star(A,B)$, not $\fU(A,B)$, so we only require independence when $A\cap B$ is a clique in $\cG_A$, that is, is maximal in $\cG_A$; it is already complete because $(A,B)$ is a decomposition. Obviously, by symmetry, $\fU^\star(A,B)$ could be defined with $A$ and $B$ interchanged without changing the meaning, but it is not the same as conditioning on the set of decomposable graphs for which $(A,B)$ is a decomposition, and $A\cap B$ is a clique in \emph{at least one of} $\cG_A$ and $\cG_B$, since in the conditional independence statement, it is $\cG$ that is random, not $(A,B)$.

The weak structural Markov property is illustrated in Figure \ref{fig:wsm}.
\begin{figure}[ht]
\begin{center}
\resizebox{\textwidth}{!}{\includegraphics{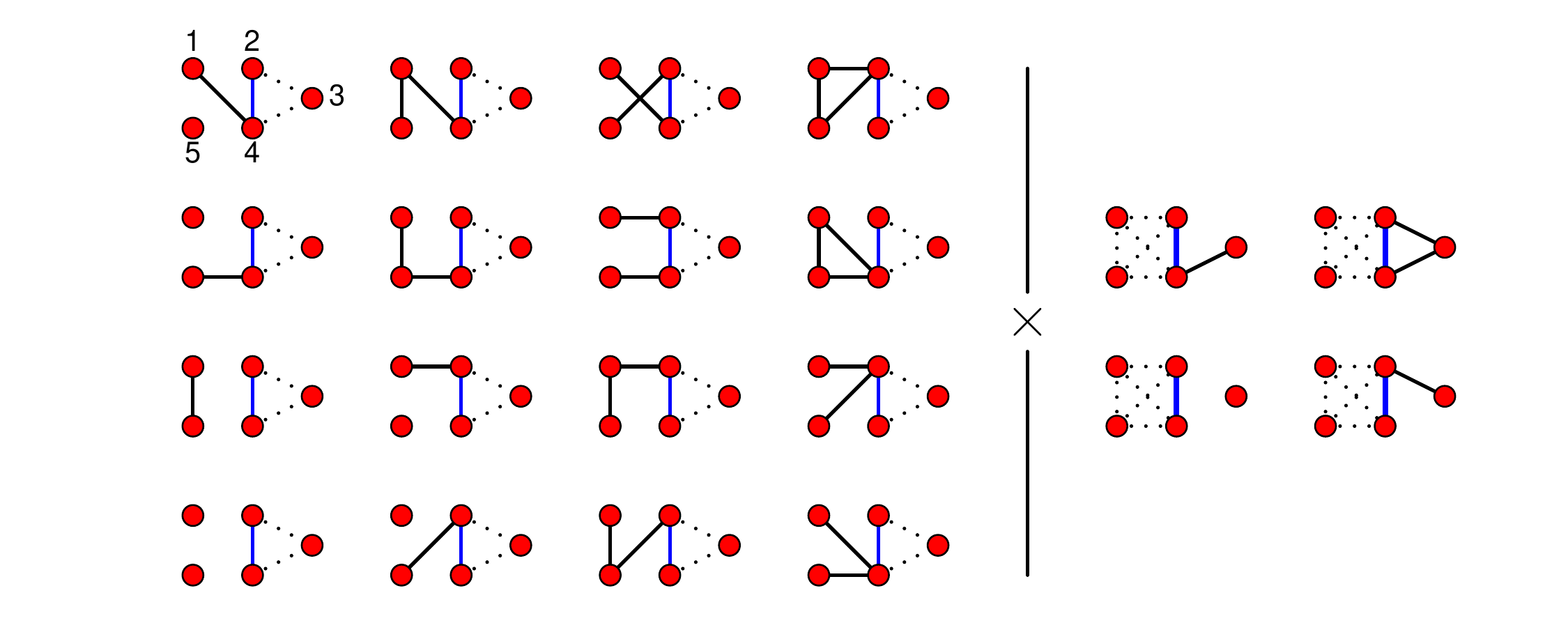}}
\end{center}
\caption{Illustrating the weak structural Markov property on a 5-vertex graph. Numbering the 5 vertices clockwise from the top left, $A$ and $B$ consist of the vertices $\{1,2,4,5\}$, and $\{2,3,4\}$ respectively. For the graph $\cG$ to lie in $\fU^\star(A,B)$, the edge $(2,4)$ must always be present, and neither $1$ and $5$ can be connected to both $2$ and $4$. The graph law must choose independently among the 16 possibilities for $\cG_A$ on the left and 
the 4 possibilities for $\cG_B$ on the right. Under the structural Markov property, there are 16 further possibilities for $\cG_A$, obtained by allowing $1$ and $5$ to be connected to both $2$ and $4$, and choice among all 32 on the left and the 4 on the right must be made independently.\label{fig:wsm}}
\end{figure}

\subsection{Clique--separator exponential family}
We now define an important family of graph laws, by an algebraic specification. This family has previously been described, though not named, by \citet{bornn-caron-11}. These authors do not examine any Markov properties of the family, but advocate it for flexible prior specification.
 
\begin{definition} (Clique--separator exponential family and clique--separator factorisation laws)
The \emph{clique--separator exponential family} is the exponential family of graph laws over $\fF \subseteq \fU$, with $(t^+,t^-)$ as natural statistic with respect to the uniform measure on $\fU$, where $t^+_A=\max(t_A,0)$ and $t^-_A=\min(t_A,0)$, and
$$
t_A(\cG)=\begin{cases}
1, & \text{ if }A\text{ is a clique in }\cG, \\
-\nu_A(\cG), &  \text{ if }A\text{ is a separator in }\cG, \\
0, & \text{otherwise},
\end{cases}
$$
where $\nu_A(\cG)$ is the multiplicity of separator $A$ in $\cG$.
That is, laws in the family have densities of the form:
$$
\pi(\cG) \propto \exp\{\omega^+ t^+(\cG)+\omega^- t^-(\cG)\}
$$
where $\omega^+=(\omega^+_A: A\subseteq V)$ and $\omega^-=(\omega^-_A: A\subseteq V)$ are real-valued set-indexed parameters, $t^+(\cG)=(t_+A(\cG): A\subseteq V)$ and $t^-(\cG)=(t_A^-(\cG): A\subseteq V)$. Here all vectors indexed by subsets of $V$ are listed in a fixed but arbitrary order, and the product of two such vectors is the scalar product. Note that $t^+_A(\cG)$ is simply 1 if $A$ is a clique in $\cG$, 0 otherwise, while $t^-_A(\cG)$ is $-\nu_A(\cG)$ if $S$ is a separator in $\cG$, again otherwise 0.
This density $\pi$ can be equivalently written as a \emph{clique--separator factorisation law}
\bel{eq:csf}
\pi_{\phi,\psi}(\cG) \propto \frac{\prod_{C \in \cC} \phi_C}{\prod_{S \in \cS}\psi_S}
\ee
where $\cC$ is the set of cliques and $\cS$ the multiset of separators of $\cG$, and
$\phi_C=\exp(\omega^+_C)$ and $\psi_S=\exp(\omega^-_S)$; this is the form we prefer to use hereafter.
\end{definition}

This definition is an immediate generalisation of that of the clique exponential family of \citet{byrne-dawid-15}, in which $t=t^++t^-$ is the natural statistic, so $\omega^+_A$ and $\omega^-_A$ coincide, as do $\phi_A$ and $\psi_A$. \citet{byrne-dawid-15} show that for any fixed vertex set, the structurally Markov laws are precisely those in a clique exponential family. In the next section we show an analogous alignment between the weak structural Markov property and clique--separator factorisation laws.
 
\subsection{Main result}
\begin{theorem} 
A graph law $\fG$ over the set $\fU$ of undirected decomposable graphs on $V$, whose support is all of $\fU$, is weakly structural Markov if and only if it is a clique--separator factorisation law.
\end{theorem}

\begin{remark} 
Exactly as in \citet[Theorem 3.15]{byrne-dawid-15} it is possible to weaken the condition of full support, that is, positivity of the density $\pi$. It is enough that if $\cG$ is in the support, so is $\cG^{(C)}$ for any clique $C$ of $\cG$.
\end{remark}

Our proof makes use of a compact notation for decomposable graphs, and a kind of ordering of cliques that is more stringent than perfect ordering/enumeration.

A decomposable graph is determined by its cliques. We write $\cGx{C_1,C_2,\ldots}$ for the decomposable graph with cliques $C_1,C_2,\ldots$. Without ambiguity we can omit singleton cliques from the list. In case the vertex set $V$ of the graph is not clear from the context, we emphasise it thus: $\cGx{C_1,C_2,\ldots}_V=\cGx{C_1,C_2,\ldots}$.
In particular, $\cGx{A}$ is the graph on $V$ that is complete in $A$ and empty otherwise, and $\cGx{A,B}$ is the graph on $V$ that is complete on both $A$ and $B$ and empty otherwise.

Recall that, starting from a list of the cliques, we can place these in a perfect sequence and simultaneously construct a junction tree by
maintaining two complementary subsets: 
those cliques visited and those unvisited. We initialize
the process by placing an arbitrary clique in the visited set and all others in the
unvisited. At each successive stage, we move one unvisited clique into the visited
set choosing arbitrarily from those that are available, that is, are adjacent
to a visited clique in the junction tree; at the same time a new link is added to the junction tree.

\begin{definition}
If at each step $j$ we select an available clique, numbering it $C_j$, such that the separator
$S_j = C_j \cap \bigcup_{i<j} C_i$ is not a proper subset of any other separator
that would arise by choosing a different available clique then we call
the ordering {\em pluperfect}.
\end{definition}

Clearly, it is computationally convenient and sufficient, but not necessary, to
choose the available clique that creates one of the largest of the separators, 
a construction closely related to the maximum cardinality search of \citet{tarjan-yannakakis-1984}.
This shows that a pluperfect ordering always exists and that any clique can be
chosen as the first.  


\begin{lemma}
Let $\pi$ be the density of a weakly structurally Markov graph law on $V$, and let $\cG$ be a decomposable graph on $V$. Consider a particular pluperfect ordering $C_1, \ldots, C_J$ of the cliques of $\cG$, and a junction tree in which the links connect $C_j$ and $C_{h(j)}$ via separator $S_j$ for each $j=2,\ldots,J$, where $h(j)\leq j-1$. For each such $j$, let $R_j$ be any subset of $C_{h(j)}$ that is a proper superset of $S_j$.
Then for any choice of such $\{R_j\}$, we have
$$
\pi(\cG)
=\prod_j \pi(\cGx{C_j}) \times \prod_{j\geq 2} \frac{\pi(\cGx{R_j,C_{j}})}{\pi(\cGx{R_j})\pi(\cGx{C_j})}
$$
\end{lemma}
\begin{proof}
Let $B=\cup_{i=1}^{j-1} C_i$. Set $A=(V\setminus B) \cup R_j$. Then $R_j\cap C_j=S_j$, $(A,B)$ is a decomposition, and $A\cap B= R_j$. This intersection $R_j$ is a clique in $\cG_A$. For, suppose for a contradiction 
that $R_j$ is not a clique in $\cG_A$, i.e., it is not maximal. Then there exists a 
vertex $v$ in $A \setminus R_j$, such that $R'=R_j \cup \{v\}$ is complete. So $R'$ is a 
subset of a clique in the original graph $\cG$. Either all the cliques containing 
$R'$ are among $\{C_i, i<j\}$, so that $v$ is not in $A$, a contradiction, or one of 
them, say $C_\star$, is among $\{C_i,i\geq j\}$, in which case there is a path in the junction tree between $C_{h(j)}$ and $C_\star$, with every clique along the path containing $R_j$; so there must be a separator that
is a superset of $R_j$ (so a strict superset of $S_j$),
connects to $C_{h(j)}$, and
is among $\{S_{j+1}, S_{j+2}, ..., S_J\}$. This
contradicts the assumption that the ordering is pluperfect.
\par\indent
This choice of $(A,B)$ forms a covering pair and $\cG \in \fU^\star(A,B)$, so under WSM, we know that $\cG_A$ and $\cG_B$ are independent under $\pi_{A,B}$, their joint distribution given that $A\cap B$ is complete in $\cGx{C_1,C_2,\ldots,C_j}$. Thus we have the cross-over identity
$$
\pi_{A,B}(\cGA{R_j},\cGB{R_j})
\times
\pi_{A,B}(\cGA{R_j,C_j},\cGB{C_1,\ldots,C_{j-1}})
=
\pi_{A,B}(\cGA{R_j},\cGB{C_1,\ldots,C_{j-1}})
\times
\pi_{A,B}(\cGA{R_j,C_j},\cGB{R_j})
$$
or equivalently,
$$
\pi(\cGx{R_j})\pi(\cGx{C_1,\ldots,C_j})=\pi(\cGx{C_1,\ldots,C_{j-1}})\pi(\cGx{R_j,C_{j}}).
$$

We can therefore write
\begin{align*}
\pi(\cG) &= \pi(\cGx{C_1}) \prod_{j\geq 2} \frac{\pi(\cGx{C_1,\ldots,C_{j}})}{\pi(\cGx{C_1,\ldots,C_{j-1}})}=
\pi(\cGx{C_1}) \prod_{j\geq 2} \frac{\pi(\cGx{R_j,C_{j}})}{\pi(\cGx{R_j})} \\
& =\prod_j \pi(\cGx{C_j}) \times \prod_{j\geq 2} \frac{\pi(\cGx{R_j,C_{j}})}{\pi(\cGx{R_j})\pi(\cGx{C_j})}
\end{align*}
\mbox{}
\end{proof}

\begin{lemma}
Let $\pi$ be the density of a weakly structurally Markov graph law on $V$, and let $S$ be any subset of the vertices $V$ with $|S|\leq n-2$.  Then $\pi(\cGx{R_1,R_2})/\{\pi(\cGx{R_1})\pi(\cGx{R_2})\}$ depends only on $S$, for all sets of vertices $R_1,R_2$ for which $R_1\cup R_2\subseteq V$, $R_1\cap R_2=S$, and where both $R_1$ and $R_2$ are strict supersets of $S$.
\end{lemma}

\begin{proof}
$\cGx{R_1,R_2}$ is a decomposable graph whose unique junction tree has cliques $R_1$ and $R_2$, and separator $S$. Applying Lemma 1 to this graph, we have 
$$
\pi(\cGx{R_1,R_2})=\pi(\cGx{R_1}) \pi(\cGx{R_2}) \times \frac{\pi(\cGx{R,R_2})}{\pi(\cGx{R})\pi(\cGx{R_2})},
$$
that is,

$$ 
\frac{\pi(\cGx{R_1,R_2})}{\pi(\cGx{R_1}) \pi(\cGx{R_2})}=\frac{\pi(\cGx{R,R_2})}{\pi(\cGx{R})\pi(\cGx{R_2})},
$$
for any $R$ with $S\subset R\subseteq R_1$. This means that any vertices may be added to or removed from $R_1$, or by symmetry to or from $R_2$, without changing the value of 
$\pi(\cGx{R_1,R_2})/\{\pi(\cGx{R_1}) \pi(\cGx{R_2})\}$,
 providing it remains true that $R_1\cup R_2\subseteq V$, $R_1\cap R_2=S$, $R_1\supset S$ and $R_2\supset S$.

But any unordered pair of subsets $R_1$, $R_2$ of $V$ with $R_1\cup R_2\subseteq V$, $R_1\cap R_2=S$, $R_1\supset S$ and $R_2\supset S$ can be transformed stepwise to any other such pair by successively adding or removing vertices to or from one or other of the subsets. Thus $\pi(\cGx{R_1,R_2})/\{\pi(\cGx{R_1})\pi(\cGx{R_2})\}$ can depend only on $S$: we will denote it by $1/\psi_S$.
\end{proof}
\begin{proof}[of Theorem 1]
Suppose that $\pi$ is the density of a weakly structurally Markov graph law on $V$. For each $A\subseteq V$, let $\phi_A=\pi(\cGx{A})$. Then by Lemmas 1 and 2,
$$ 
\pi(\cG) =\frac{\prod_j \phi_{C_j}}{\prod_{j \geq 2} \psi_{S_j}}.
$$
Since $\cGx{\emptyset}$, $\cGx{\{v\}}$, $\cGx{\{w\}}$, $\cGx{\{v,w\}}$, for distinct vertices $v,w\in V$ all denote the same graph, we must have $\phi_{\{v\}}=\pi(\cGx{\emptyset})$ for all $v$, and also $\psi_{\emptyset}=\pi(\cGx{\emptyset})$. Under these conditions, the constant of proportionality in (\ref{eq:csf}) is evidently 1. 

Conversely, it is trivial to show that if the clique--separator factorisation property (\ref{eq:csf}) applies to $\pi$, then $\pi$ is the density of a weakly structurally Markov graph law.
\end{proof}

\subsection{Identifiability of parameters}
\label{sec:ident}
\citet[Proposition 3.14]{byrne-dawid-15} point out that their $\{t_A(\cG)\}$ values are subject to $|V|+1$ linear constraints, $\sum_{A\subseteq V} t_A(\cG)=1$, $\sum_{A\ni v} t_A(\cG)=1$ for all $v\in V$, so that their parameters $\omega_A$, or equivalently $\phi_A$, are not all identifiable. They obtain identifiability, by proposing a standardised vector $\omega^\star$, with $|V|+1$ necessarily 0 entries, that is a linear transform of $\omega$. By the same token, the $|V|+1$ constraints on $t_A(\cG)$ are linear constraints on $t^+_A(\cG)$ and $t^-_A(\cG)$, and so $\{\phi_A\}$ and $\{\psi_A\}$ are not all identifiable. We could obtain identifiable parameters by for example choosing $\psi_{\emptyset}=1$ and $\phi_{\{v\}}=1$ for all $v \in V$, or, as above, by setting $\psi_{\emptyset}=\pi(\cGx{\emptyset})$ and $\phi_{\{v\}}=\pi(\cGx{\emptyset})$ for all $v$, or in other ways. 

Note in addition that $\emptyset$ cannot be a clique, and neither $A=V$ nor any subset $A$ of $V$ with $|A|=|V|-1$ can be a separator, so the corresponding $\phi_A$ and $\psi_A$ are never used. The dimension of the space of clique--separator factorisation laws is therefore $2 \times 2^{|V|}-2|V|-3$, nearly twice that of clique exponential family laws, $2^{|V|}-|V|-1$.

For example, when $|V|=3$, all graphs are decomposable, and all graph laws are clique--separator factorisation laws, while clique exponential family laws have dimension 4; when $|V|=4$, 61 out of 64 graphs are decomposable, and the dimensions of the two spaces of laws are 21 and 11; when $|V|=7$, only 617675 out of the $2^{21}$ graphs are decomposable, and the dimensions are 239 and 120.

\begin{figure}[ht]
\begin{center}
\resizebox{\textwidth}{!}{\includegraphics{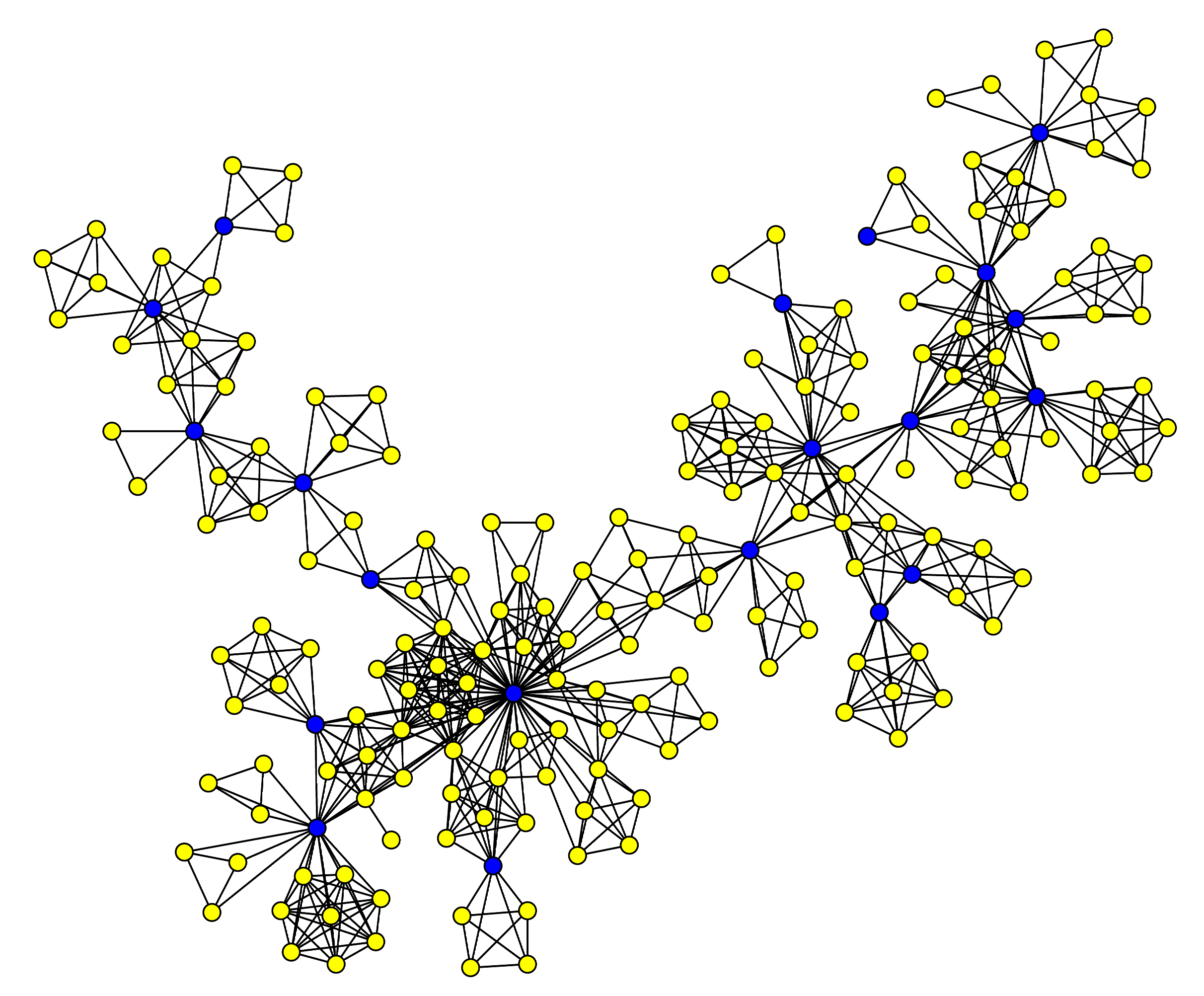}}
\end{center}
\caption{Simulated graph from a clique--separator factorisation model, with $\phi_C=\exp(-4|C|)$; $\psi_S=\exp(-0.5|S|)$ if $S$ contains a hub (dark colour), otherwise $\infty$. There are 200 vertices including 20 hubs.\label{fig:universal}}
\end{figure}

\section{Some implications for modelling and inference}

\subsection{Conjugacy and posterior updating}
As priors for the 
graph underlying
a model $P(X|\cG)$ for data $X$, clique--separator factorisation laws are conjugate for decomposable likelihoods, in the case where there are no unknown parameters in the distribution: given $X$ from the model
$$
p(X|\cG) = \frac{\prod_{C\in \cC} \lambda_C(X_C)}{\prod_{S\in \cS} \lambda_S(X_S)}
=\prod_{A\subseteq V} \lambda_A(X_A)^{t_A(\cG)}
$$
where $\lambda_A(X_A)$ denotes the marginal distribution of $X_A$,
the posterior for $\cG$ is
$$
p(\cG|X) \propto \frac{\prod_{C\in \cC} \phi_C\lambda_C(X_C)}{\prod_{S\in \cS} \psi_S\lambda_S(X_S)}
$$
that is, a clique--separator factorisation law with parameters $\{\phi_A\lambda_A(x_A)\}$ and $\{\psi_A\lambda_A(x_A)\}$.

More generally, when there are parameters in the graph-specific likelihoods, the notions of compatibility and hyper-compatibility \citep{byrne-dawid-15} allow the extension of the idea of structural Markovianity to the joint Markovianity of the graph and the parameters, and gives the form of the corresponding posterior.

\subsection{Computational implications}
Computing posterior distributions of graphs on a large scale remains problematic, with Markov chain Monte Carlo methods seemingly the only option except for toy problems, and these methods having notoriously poor mixing. However, the junction tree sampler of \citet{green-thomas-13} seems to give acceptable performance for moderate-sized problems of up to a few hundred variables. Posteriors induced by clique--separator factorisation law priors are ideal material for these samplers, which explicitly use a clique--separator representation of all graphs and distributions. 

In \citet{bornn-caron-11}, a different Markov chain Monte Carlo sampler for clique--separator factorisation laws is introduced. We have evidence that the examples shown in their figures are not representative samples from the particular models claimed, due to poor mixing.

\subsection{Modelling}

Here we briefly discuss
the way in which choice of particular forms for the parameters $\phi_A$ and $\psi_A$ govern the qualitative and even quantitative aspects of the graph law. These choices are important in designing a graph law for a particular purpose, whether or not this is prior modelling in Bayesian structural learning.

A limitation of clique exponential family models is that because large clique potentials
count in favour of a graph, and large separator potentials count against,
it is difficult for these laws to encourage the same features
in both cliques and separators. For instance, if we choose clique potentials
to favour large cliques, we seem to be forced to favour small separators.

A popular choice for a graph prior in past work on Bayesian structural learning is the well-known Erd\H{o}s--R\'enyi random graph model, in which each of the $|V|(|V|-1)/2$ possible edges on the vertex set $V$ is present independently, with probability $p$. This model is amenable to theoretical study, but realisations of this model typically exhibit no discernible structure. When restricted to decomposable graphs, the Erd\H{o}s--R\'enyi model is a 
rather extreme example of a clique exponential family law, arising by taking $\phi_A=(p/(1-p))^{|A|(|A|-1)/2}$. Again realisations appear unstructured, essentially because of the quadratic dependence on clique or separator size in the exponent of the potentials $\phi_A$.

For a concrete example of a model with much more structure, suppose that our decomposable graph represents a communication
network. There are two types of vertices, hubs and
non-hubs. Adjacent vertices can all communicate with each other,
but only hubs will relay messages. So, for a non-hub to communicate
with a non-adjacent non-hub, there must be a path in the graph from
one to the other where all intermediate nodes are hubs. This example has the interesting feature that using only local properties, it enforces
a global property, universal communication. 
A necessary and sufficient condition for universal communication
is that every separator contains a hub. This implies that either the graph
is a single clique, or every clique must also contain a hub.
To model this with a clique--separator factorisation law, we can set the separator potential to
be $\psi_S= \infty$ if $S$ does not contain a hub. We are free to set the
remaining values of $\psi_S$, and the values of
the clique potentials $\phi_C$ for all cliques $C$, as we wish. In this example,
these parameters are chosen to control the sizes of cliques and separators; specifically, $\phi_C=\exp(-4|C|)$ and $\psi_S=\exp(-0.5|S|)$ when $S$ contains a hub, which discourages both
large cliques, and separators containing only hubs. The graph probability $\pi(\cG)$ will be
zero for all decomposable graphs that fail to allow universal communication,
and otherwise will follow the distribution implied by the potentials. Note that this example
requires the slight generalisation of the Theorem 1 mentioned in the Remark following it.
Figure \ref{fig:universal} shows a sample from this model, generated using a junction tree sampler.

\subsection{Significance for statistical analysis}

This short paper is not the place for a comprehensive investigation of the practical implications of adopting prior models from the clique--separator factorisation family in statistical analysis, something we intend to explore in later work. Instead, we extend the discussion of the example of the previous section to draw some lessons about inference. First we make the simple but important observation that the support of the posterior distribution of the graph cannot be greater than that of the prior. So, in the example of the hub model, the posterior will be concentrated on decomposable graphs where every separator contains a hub, and realisations will have some of the character of Figure \ref{fig:universal}.

There has been considerable interest recently in learning graphical models using methods that implicitly or explicitly favour hubs, defined in various ways with some affinity to our use of the term; see for example, 
\citet{mohan2014node},
\citet{tan2014learning} and
\citet{zhang2017incorporating}. These are often motivated by genetic applications in which hubs may be believed to correspond to genes of special significance in gene regulation. These methods usually assume that the labelling of nodes as hubs is unknown, but it is straightforward to extend our hub model to put a probability model on this labelling, and to augment the Monte Carlo posterior sampler with a move that reallocates the hub labels, using any process that maintains the presence of at least one hub in every separator. This is a strong hint of the possibility of a fully Bayesian procedure that learns graphical models with hubs.

\subsection{The cost of assuming the graph is decomposable when it is not}

The assumption of a decomposable graph law as prior on Bayesian structural learning is of course a profound restriction. There is no reason why nature should have been kind enough to generate data from 
graphical models that 
are decomposable. However the computational advantages of such an assumption are tremendous; see the experiments and thorough review in \citet{jones-carvalho-etal-05}. The position has not changed much since this paper was written, so far as computation of exact posteriors is concerned.

However, an optimistic perspective on this conflict between prior reasonableness and computation tractability can be justified by work of \citet{fitch-jones-etal-14}. For the zero-mean Gaussian case, with a hyper-inverse-Wishart prior on the concentration matrix, they conclude that asymptotically
the posterior will converge to graphical structures that are minimal
triangulations of the true graph,
the marginal log likelihood ratio comparing different minimal
triangulations is stochastically bounded and appears to remain data
dependent regardless of the sample size, and
the covariance matrices corresponding to the different minimal
triangulations are essentially equivalent, so model averaging is of minimal
benefit.
Informally, restriction to decomposable graphs doesn't matter really, with the right parameter priors; we can still fit essentially the right model, though perhaps inference on the graph itself should not be over-interpreted.

\section{An even weaker structural Markov property}
It is tempting to wonder if clique--separator factorisation is equivalent to a simpler definition of weak structural Markovianity, one that places yet fewer conditional independence constraints on $\fG$; the existence of the Theorem makes this possibility implausible, but it remains conceivably possible that a smaller collection of conditional independences could be equivalent. The following counter-example rules out the possibility of requiring only that 
\bel{eq:ewsm}
\cG_A \CI \cG_B \mid \{\cG \in \fU^{+}(A,B)\} \quad[\fG],
\ee
where $\fU^{+}(A,B)$ is the set of decomposable graphs for which $(A,B)$ is a decomposition, and $A\cap B$ is a clique in $\cG$.

\begin{example}
Consider graphs on vertices $\{1,2,3,4\}$. The only non-trivial conditional independence statements implied by property (\ref{eq:ewsm}) arise from decompositions $(A,B)$ where both $A$ and $B$ have three vertices (and $A\cap B$ has two). Suppose $A=\{1,2,3\}=123$, for short, and $B=234$. Given that $23$ is a clique in $\cG$, $\cG_A$ may be $\cGx{23}_A$, $\cGx{12,23}_A$ or $\cGx{13, 23}_A$, 
and similarly 
$\cG_B$ may be $\cGx{23}_B$, $\cGx{23,24}_B$ or $\cGx{23,34}_B$. 
These two choices are independent, by (\ref{eq:ewsm}), and this imposes 4 equality constraints on the graph law. There are $6$ 
different choices for the two-vertex clique $A\cap B$, so not more than 24 constraints overall (they may not all be independent). There are 61 decomposable graphs on four vertices, so the set of graph laws satisfying (\ref{eq:ewsm}) has dimension at least $60-24=36$. But as we saw in section \ref{sec:ident}, the set of clique--separator factorisation laws has dimension $2 \times 2^{|V|}-2|V|-3 = 21$.
\end{example}

Essentially, assumption (\ref{eq:ewsm}) does not constrain the graph law sufficiently to obtain the explicit clique--separator factorisation. In fact, it is easy to show that (\ref{eq:ewsm}) places no constraints on $\pi(\cG)$ for any connected $\cG$ consisting of one or two cliques.

\section*{Acknowledgements}

We thank the editors and referees for their constructive criticisms of our presentation, and we are also grateful to
Luke Bornn, Simon Byrne, Fran\c{c}ois Caron, Phil Dawid and Steffen Lauritzen for stimulating discussions and correspondence.
Alun Thomas was supported in part by funding from the
National Center for Research Resources and the National Center for Advancing
Translational Sciences, National Institutes of Health,
through Grant 5UL1TR001067-05 (formerly 8UL1TR000105 and UL1RR025764). 


\end{document}